\documentclass[conference]{IEEEtran}
\IEEEoverridecommandlockouts
\usepackage{cite}
\usepackage{amsmath,amssymb,amsfonts}
\usepackage{algorithmic}
\usepackage{graphicx}
\usepackage{textcomp}
\usepackage[english]{babel}
\usepackage{dsfont}
\usepackage{xcolor}
\usepackage{autobreak}
\usepackage{enumitem}
\usepackage{amsthm}
\usepackage{siunitx}  
\usepackage{bm}
\usepackage{enumitem}

\usepackage{subfig}
\usepackage{xcolor}
\def\BibTeX{{\rm B\kern-.05em{\sc i\kern-.025em b}\kern-.08em
    T\kern-.1667em\lower.7ex\hbox{E}\kern-.125emX}}

\allowdisplaybreaks[3]
\makeatletter
\newcommand{\vast}{\bBigg@{5}}
\newcommand{\Vast}{\bBigg@{5.5}}

\newtheorem{proposition}{Proposition}
\makeatother
\begin{document}
\title{Steady-State Optimal Frequency Control for Lossy Power Grids with Distributed Communication\\
\thanks{This work was funded by the Deutsche Forschungsgemeinschaft (DFG, German Research Foundation)---project number 360464149.	
}
}
\author{\IEEEauthorblockN{Lukas K\"olsch, Kirtan Bhatt, Stefan Krebs, and S\"oren Hohmann}
	\IEEEauthorblockA{\textit{Institute of Control Systems, Karlsruhe Institute of Technology (KIT)}, Karlsruhe, Germany \\
		lukas.koelsch@kit.edu, kirtan.bhatt@student.kit.edu, stefan.krebs@kit.edu, soeren.hohmann@kit.edu}
}
\maketitle
\begin{abstract}
We present a distributed and price-based control approach for frequency regulation in power grids with nonzero line conductances. Both grid and controller are modeled as a port-Hamiltonian system, where the grid model consists of differential as well as algebraic equations.
Simulations show that the resulting controller asymptotically stabilizes the frequency while maintaining minimum overall generation costs in steady state and being robust in terms of clock drifts and uncontrollable loads. Moreover, it is shown that active power sharing can be achieved by an appropriate choice of the cost function.
\end{abstract}
\section{Introduction}
\subsection{State of Research}
Frequency control and thus regulation of the balance between generation and consumption in the electrical grid has so far been the task of the transmission system operator.
With the worldwide trend towards more renewable energy generation and a displacement of large conventional power plants, there is an increasing number of small-scale generation, 
which raises the need for a replacement of the centralized control strategy by a distributed one.
Distributing frequency regulation to several agents allows to divide a complex task into several smaller tasks which are solved in parallel by individual agents. In addition, distributed control results in an increased robustness of the overall system with respect to preventing single points of failures and attacks from outside \cite{Patnaik.2018}.

A class of distributed control concepts that has recently been very popular in terms of frequency regulation and balancing is real-time dynamic pricing (see \cite{Doerfler2019} for a detailed survey on current research directions).
Dynamic pricing is particularly advantageous in large scale networks as it enables implicit communication of momentary imbalances via a price signal, resulting in a dynamic feedback minimization \cite{Jokic.2007,Colombino2018,Lawrence2018,Menta2018} of the overall costs. The price signal represents an aggregated information about current imbalances between generations and consumptions. Thus the control can be distributed based on neighbor-to-neighbor communication as well as local measurements and local control. The actual status of the network does not have to be completely known to the individual agents. 

Based on the formulation of a specific overall cost function $C$, the distributed minimization of this cost function allows certain goals to be achieved at the same time, such as equal marginal prices for each agent or active power sharing.


In previous publications dealing with real-time dynamic pricing for frequency control of power grids, 
the controllers were always designed under the assumption that line conductances are all zero \cite{Mallada.2017,Zhao2018,Köhler.2017, MohsenianRad.2010,Trip.2016,Stegink.}. However, this is an inadmissible assumption especially for distribution grids, see. e.g. \cite{Marano2012}. In fact it can be shown by simulation that applying these controllers to the lossy AC power flow model always leads to a synchronous frequency $\overline \omega$ which deviates from the nominal frequency $\omega^{n}$. As a consequence, a practical implementation of all of these controllers would still require some kind of additional frequency restoration or secondary frequency control by a TSO. 
Thus a further development of distributed control algorithms is needed to provide a benefit compared to the classical hierarchy of primary, secondary, and tertiary frequency control in the lossy case.

\subsection{Main Contributions}
To overcome the steady-state frequency deviation mentioned above, we propose an extended price-based and distributed controller, which takes into account the local nonzero conductances of neighboring lines and leads to zero deviation from nominal frequency $\omega^n$. To facilitate transient stability analysis, we represent both plant and controller as a port-Hamiltonian system, which results in a closed-loop system that is again port-Hamiltonian.
Stability is then derived from a shifted passivity property with respect to the post-fault equilibrium.
 
We compare our approach with the one presented in \cite{Stegink.} using a simulation example with both controllable and uncontrollable infeeds and loads. We show that stability of the closed-loop system is given under various communication topologies. Furthermore, we show the transient behavior of the closed loop under step load changes, incorrect measurements and communication failures.

The remainder of this paper is organized as follows. 
In Section II, we give some notational remarks and introduce a port-Hamiltonian formulation for a power grid with generator and load nodes which are coupled via lossy AC power lines.
In Section III, we deploy a distributed price-based controller which aims at minimizing overall generation costs while keeping the steady-state synchronous frequency to the nominal frequency.
In Section IV, we assess our approach on a test network with variable load scenarios as well as different communication topologies. 
Finally, in Section V, we sum up our contributions and discuss trends for future work.
\section{Power Grid Model} \label{sec-PGM}

\subsection{Notational Preliminaries}
Positive semidefiniteness of a matrix is denoted by $\succeq 0$, whereas element-wise nonnegativity of a vector or matrix is denoted by $\geq 0$. Vector $\bm a = \mathrm{col}_i\{a_i\}=\mathrm{col}\{a_1,a_2,\ldots\}$ is a column vector of elements $a_i$, $i=1,2,\ldots$ and matrix $\bm A=\mathrm{diag}_i\{a_i\}=\mathrm{diag}\{a_1,a_2,\ldots\}$ is a (block-)diagonal matrix of elements $a_i$, $i =1,2,\ldots$. The $(n \times n)$-identity matrix and $(n \times n)$-zero matrix are denoted by $\bm I_n$ and $\bm 0_n$, respectively. For all other vectors and matrices, the dimensions are either explicitly specified or they result from the context.

The power grid is modeled by a directed graph $\mathcal G_p=(\mathcal V, \mathcal E_p)$ with $\mathcal V = \mathcal V_g \cup \mathcal V_\ell$ being the set of $n_g=|\mathcal V_g|$ generator nodes and $n_\ell=|\mathcal V_\ell|$ load nodes, respectively. The physical interconnection is represented by the incidence matrix $\bm D_p \in \mathds R^{n\times m_p}$ with $n=n_g+n_\ell$ and $m_p = |\mathcal E_p|$. Incidence matrix $\bm D_p$ can be subdivided as follows
\begin{align}
\bm D_p = \begin{bmatrix} \bm D_{pg} \\ \bm D_{p \ell} \end{bmatrix},
\end{align}
where submatrices $\bm D_{p g}$ and $\bm D_{p \ell}$ correspond to the generator and load nodes, respectively.
 
We note $j \in \mathcal N_i$ if node $j$ is a \emph{neighbor} of node $i$, i.e. $j$ is adjacent to $i$ in the undirected graph.

A list of all parameters and variables of the power grid is given in Table \ref{tab1}.
\subsection{Modeling Assumptions} \label{sec-assumptions}
Following the lines of \cite{Stegink.2017} and \cite{Stegink.}, we make the following modeling assumptions for power grid model and controller:
\begin{enumerate}
	\item The grid is operating around the nominal frequency $\omega^n=2\pi \cdot \SI{50}{\hertz}$.
	\item The grid is a balanced three-phased system and the lines are represented by its one-phase $\pi$-equivalent circuits.
	\item Subtransient dynamics of the synchronous generators are neglected.
	\item Delays in communication as well as sensor and actor delays of the controllers are neglected.
\end{enumerate}
However, we make the following additional (relaxed) assumptions:
\begin{enumerate}[resume]
	\item Power lines are lossy, i.e. have nonzero conductances.
	\item Loads do not have to be constant.
	\item Excitation voltages of the generators do not have to be constant.
\end{enumerate}
\begin{table}[!t]
					\renewcommand{\arraystretch}{1.3}
	\caption{List of Parameters and State Variables of Power System}
	\begin{center}
		\begin{tabular}{|l||l|}
			\hline
			$A_i$ & positive generator and load damping constant \\
			$B_{ij}$ & negative of susceptance of line $(i,j)$ \\
			$\bm D_p$ & incidence matrix of power grid\\
			$G_{ij}$ & negative of conductance of line $(i,j)$\\
			$L_i$ & deviation of angular momentum from the nominal value $M_i\omega^n$ \\
			$M_i$ & moment of inertia \\ 
			$p_i$ & sending-end active power flow \\
			$p_{g,i}$ & active power generation \\
			$p_{\ell,i}$ & active power demand \\
			$q_i$ & sending-end reactive power flow \\
			$q_{\ell,i}$ & reactive power demand \\
			$U_i$ & magnitude of transient internal voltage \\
			$U_{f,i}$ & magnitude of excitation voltage \\
			$X_{d,i}$ & d-axis synchronous reactance \\
			$X'_{d,i}$ & d-axis transient reactance \\
			$\theta_i$ & bus voltage phase angle \\
			$\vartheta_{ij}$ & bus voltage angle difference $\theta_i-\theta_j$\\
			$\Phi$ & overall transmission losses \\
			$\tau_{U,i}$ & open-circuit transient time constant of the synchronous machine \\
			$\omega_i$ & deviation of bus frequency from the nominal value $\omega^n$ \\
			\hline
		\end{tabular}
		\label{tab1}
	\end{center}
\end{table}
\subsection{Dynamic Model of Generator Buses}
As in \cite{Trip.2016}, each generator node is represented by a third-order synchronous generator model in local dq coordinates: 
\begin{align}
\dot \theta_i &= \omega_i, && i \in \mathcal V_g, \label{plant_first_eq}\\
\dot L_i &= -A_i \omega_i + p_{g,i}-p_{\ell,i}-p_i, && i \in \mathcal V_g, \\
\tau_{U,i}\dot U_i &= U_{f,i}-U_i - (X_{d,i}-X'_{d,i})U_i^{-1} \cdot q_i, && i \in \mathcal V_g.
\end{align}
Without loss of generality, yet for ease of notation, we assume that each generation $p_{g,i}$ is controllable and each load $p_{\ell,i}$ and $q_{\ell,i}$ is uncontrollable. 
Controllable loads can easily be added as negative generations. Thus $p_{g,i}$ are the control inputs while $p_{\ell,i}$ and $q_{\ell,i}$ act as disturbance inputs.
\subsection{Dynamic Model of Load Buses}
The load nodes are supposed to be uncontrollable and to have a frequency-dependent active power consumption, which is modeled by load damping coefficients $A_{i}\geq 0$ and which leads to the following set of differential-algebraic equations:
\begin{align}
\dot \theta_i &= \omega_i, && i \in \mathcal V_\ell, \\
0 &= -A_i \omega_i -p_{\ell,i}-p_i, && i \in \mathcal V_\ell, \\
0 &= -q_{\ell,i}-q_i, && i \in \mathcal V_\ell.
\end{align}
\subsection{Power Line Model}
The (sending-end) active and reactive power flows of node $i \in \mathcal V$ are given by the lossy AC power flow equations \cite{Machowski.2012}
\begin{align}
p_i &= \sum_{j \in \mathcal N_i} B_{ij}U_iU_j \sin(\theta_i - \theta_j) + G_{ii}U_i^2 && \nonumber\\
	&+ \sum_{j \in \mathcal N_i}G_{ij}U_iU_j\cos(\theta_i-\theta_j),&& i \in \mathcal V, \\
q_i &= -\sum_{j \in \mathcal N_i} B_{ij}U_iU_j \cos(\theta_i - \theta_j) + B_{ii}U_i^2 && \nonumber\\
&+ \sum_{j \in \mathcal N_i}G_{ij}U_iU_j\sin(\theta_i-\theta_j),&& i \in \mathcal V \label{plant_last_eq}
\end{align}
with $\bm Y = \bm G + \mathrm j \bm B$ being the admittance matrix. Note that by definition of the admittance matrix, $G_{ij}<0$ and $B_{ij}>0$ if nodes $i$ and $j$ are connected via a resistive-inductive line \cite{Machowski.2012}.
\subsection{Overall Model}
In order to get a port-Hamiltonian state space model of the \emph{plant}, i.e. the open-loop system, we chose the plant state vector as follows
\begin{align}
\bm x_p = \mathrm{col}\{\bm \vartheta, \bm L, \bm U_g, \bm \omega_{\ell,i}, \bm U_\ell\}
\end{align}
where for convenience we define the angle deviation $\bm \vartheta = \bm D_p^\top \bm \theta$ and the angular momenta $\bm L = \mathrm{col}_i\{L_i\}$ with $L_i = M_i\cdot \omega_i$, $i \in \mathcal V_g$.

To describe the energy stored in the power network, we choose the following positive-definite function as the plant Hamiltonian
\begin{align}
H_p(\bm x_p) &= \frac 12 \sum_{i \in \mathcal V_g}\left( M_i^{-1}L_i^2 + \frac{U_i^2}{X_{d,i}-X'_{d,i}}\right) \nonumber \\
&- \frac 12 \sum_{i \in \mathcal V_g \cup \mathcal V_\ell} B_{ii}U_i^2 - \sum_{(i,j) \in \mathcal E}B_{ij}U_iU_j \cos(\theta_i - \theta_j) \nonumber \\
&+\frac 12 \sum_{i \in \mathcal V_\ell} \omega_{\ell,i}^2
\end{align}
where the first row represents the (shifted) kinetic energy of the rotors and the magnetic energy of the generator circuits, the second row represents the magnetic energy of the transmission lines and the third row represents the local deviations of the loads from nominal frequency, i.e. the deviation from an ideal grid-supporting behavior.

Combining \eqref{plant_first_eq}--\eqref{plant_last_eq}, 
we get the following port-Hamiltonian formulation for the power grid:
\begin{align}
\begin{bmatrix}
\dot{\bm \vartheta} \\ \dot{\bm L} \\ \dot{\bm U}_g \\ \bm 0 \\ \bm 0 \end{bmatrix}
=&\vast[\underbrace{\begin{bmatrix}
\bm 0 & \bm D_{pg}^\top & \bm 0 & \bm D_{p\ell}^\top & \bm 0 \\
-\bm D_{pg} & \bm 0 & \bm 0 & \bm 0 & \bm 0 \\
\bm 0 & \bm 0 & \bm 0 & \bm 0 & \bm 0 \\
-\bm D_{p\ell} & \bm 0 & \bm 0 & \bm 0 & \bm 0 \\
\bm 0 & \bm 0 & \bm 0 & \bm 0 & \bm 0
\end{bmatrix}}_{\bm J_p} \nonumber \\
&-\underbrace{\begin{bmatrix}
	\bm 0 & \bm 0 & \bm 0 & \bm 0 & \bm 0 \\
	\bm 0 & \bm{A}_g & \bm 0 & \bm 0 & \bm 0 \\
	\bm 0 & \bm 0 & \bm{R}_g & \bm 0 & \bm 0 \\
	\bm 0 & \bm 0 & \bm 0 & \bm{A}_\ell & \bm 0 \\
	\bm 0 & \bm 0 & \bm 0 & \bm 0 & \widehat{\bm U}_\ell
	\end{bmatrix}}_{\bm R_p}\vast]
\nabla H_p \nonumber \\
&-
\underbrace{\begin{bmatrix}
\bm 0 \\ \bm{\varphi_g} \\ \bm{\varrho_g} \\ \bm{\varphi_\ell} \\ \bm{\varrho_\ell}
\end{bmatrix}}_{\bm r_p}
+
\begin{bmatrix}
\bm 0 & \bm 0 & \bm 0 & \bm 0 \\
\bm I & \bm 0 & \bm 0 & -\bm{\widehat I}_g \\
\bm 0 & \bm{\hat \tau}_U & \bm 0 & \bm 0 \\
\bm 0 & \bm 0 & \bm 0 & -\bm{\widehat I}_\ell \\
\bm 0 & \bm 0 & -\bm I & \bm 0
\end{bmatrix}
\begin{bmatrix}
\bm{p}_g \\ \bm{U}_f \\ \bm{q}_\ell \\ \bm{p}_\ell  
\end{bmatrix} 
\label{plant-phs}
\end{align}
with
\begin{align}
\bm A_g &= \mathrm{ diag}_i\{A_i\},&&i \in \mathcal V_g \\
\bm A_\ell &= \mathrm{ diag}_i\{A_i\}, && i \in \mathcal V_\ell \\
\bm R_g &= \mathrm{ diag}_i\left\{\frac{X_{di}-X_{di}'}{\tau_{U,i}}\right\}, && i \in \mathcal V_g \\
\widehat{\bm U}_\ell &= \mathrm{ diag}_i\{U_i\}, && i \in \mathcal V_\ell\\
\bm \varphi_g &= \mathrm{ col}_i\big\{G_{ii}U_i^2 + \sum_{j \in \mathcal N_i}G_{ij}U_iU_j\cos(\vartheta_{ij})\big\}, && i \in \mathcal V_g \\
\bm \varphi_\ell &= \mathrm{ col}_i\big\{G_{ii}U_i^2 + \sum_{j \in \mathcal N_i}G_{ij}U_iU_j\cos(\vartheta_{ij})\big\}, && i \in \mathcal V_\ell \\
\bm \varrho_g &= \mathrm{ col}_i\big\{R_{g,i}\sum_{j \in \mathcal N_i} G_{ij}U_iU_j\sin(\vartheta_{ij})\big\},&&i \in \mathcal V_g \\
\bm \varrho_\ell &= \mathrm{ col}_i\big\{\sum_{j \in \mathcal N_i} G_{ij}U_iU_j\sin(\vartheta_{ij})\big\},&&i \in \mathcal V_\ell \\
\bm{\hat \tau}_U &= \mathrm{ diag}_i\{1 / \tau_{U,i}\}, && i \in \mathcal V_g \\
\widehat{\bm I}_g &= \begin{bmatrix} \bm I_{n_g \times n_g} & \bm 0_{n_g \times n_\ell} \end{bmatrix},\\
\widehat{\bm I}_\ell &= \begin{bmatrix} \bm 0_{n_\ell \times n_g}   &\bm I_{n_\ell \times n_\ell} \end{bmatrix}. 
\end{align}
With $\bm J_p = -\bm J_p^\top$ and $\bm R_p \succeq 0$, this is a port-Hamiltonian descriptor system \cite{Beattie2017} with a nonlinear dissipative relation \cite{vanderSchaft.2017}.
\section{Controller Design}
\subsection{Control Objective}
The aim is to minimize a certain overall generation cost function $C(\bm{p}_g)$ which is assumed to be strictly convex. 

To allow only meaningful injection profiles $\bm{p}_g$, we add the following active power balance as an additional equality constraint, since it is a necessary condition for equlibrium of \eqref{plant-phs}:
\begin{align}
\Phi&=\sum_{i \in \mathcal V_g}{p_{g,i}} - \sum_{i \in \mathcal V}{p_{\ell,i}} \nonumber \\
&= \sum_{i \in \mathcal V} G_{ii}U_i^2 + 2 \cdot \sum_{(i,j) \in \mathcal E} G_{ij} U_iU_j\cos(\vartheta_{ij}), \label{eq-balance-scalar}
\end{align}
i.e. the surplus of energy (generation minus load) must be equal to the transmission line losses.

This leads to the following constrained optimization problem:
\begin{align*}
\begin{array}{ll}
\displaystyle \min&C(\bm p_{g}) \\
\mathrm{subject\ to} &\eqref{eq-balance-scalar}
\end{array}
\tag{OP1}\label{opt-P1}
\end{align*}
Note again that controllable loads $-p_{g,i}$ with (strictly concave) utility functions can be modeled as generations with (strictly convex) generation cost functions. 

Following the lines of \cite{Stegink.2017,Stegink.}, a distributed representation of \eqref{opt-P1} is derived by transforming the scalar balance condition given by \eqref{eq-balance-scalar} into a vector comprising of $ n = n_g+n_\ell $ scalar equations.
Equation \eqref{opt-P1} is fulfilled if and only if there exists some vector $\bm \nu \in \mathds R^{m_c}$, called the vector of \textit{virtual power flows} \cite{Stegink.}, such that
\begin{align}
\bm D_c \bm \nu = \widehat{\bm I}_g^\top\bm p_g  - \bm p_{\ell} - \bm\varphi \label{eq-balance-komm}
\end{align}
with $\bm D_c$ being an arbitrary incidence matrix of a communication graph $\mathcal G_c=(\mathcal V,\mathcal E_c)$ with $m_c= |\mathcal E_c|$ edges and $\bm\varphi = \mathrm{col}\{\bm\varphi_g , \bm\varphi_\ell\}$. As will become obvious later, the adjacency relationships of the communication graph determine which generator nodes exchange variables.

A distributed representation of \eqref{opt-P1} is then given as follows. 
\begin{align*}
\begin{array}{ll}
\displaystyle \min&C(\bm p_{g}) \\
\mathrm{subject \ to} &\eqref{eq-balance-komm}
\end{array}
\tag{OP2}\label{opt-P2}
\end{align*}
Since \eqref{eq-balance-scalar} as well as \eqref{eq-balance-komm} are affine in $\bm p_{g}$, the respective optimization problems \eqref{opt-P1} and \eqref{opt-P2} are strictly convex. Moreover it can be shown 
in \cite{Stegink.2017} that \eqref{opt-P2} is an exact relaxation of \eqref{opt-P1} and thus each optimal solution of \eqref{opt-P2} is also optimal w.r.t. \eqref{opt-P1}.
\subsection{Distributed Control Algorithm}
The constrained optimal pricing \eqref{opt-P2} can be achieved by means of a distributed control algorithm for $\bm p_g$.
This is stated in the following proposition:
\begin{proposition}
Let the assumptions from \ref{sec-assumptions} be given with power grid model \eqref{plant-phs} and optimization problem \eqref{opt-P2}. Then each equilibrium of the distributed control algorithm
\begin{align}
\bm \tau_g \dot{\bm p}_{g} &= - \nabla C({\bm p}_{g})+ \widehat{\bm I}_g{\bm \lambda} + \bm u_c, \label{primal-dual-1}\\
\bm \tau_{\lambda}\dot{\bm \lambda}&= \bm D_{c} \bm \nu - \widehat{\bm I}_g^\top\bm p_{g}+  \bm p_{\ell} + \bm \varphi, \\
\bm \tau_\nu \dot{\bm \nu} &= -\bm D_{c}^\top \bm \lambda \label{primal-dual-4}
\end{align}
with $\bm \tau_g, \bm \tau_{\lambda}, \bm \tau_\nu > 0$ is a minimizer of \eqref{opt-P2}.
\end{proposition}
\begin{proof}
Provided that Slater's condition is fulfilled, a necessary condition for an optimum $ \overline{\bm x}_c :=\mathrm{col}\{\overline{\bm p}_g, \overline{\bm\lambda}, \overline{\bm\nu}\}$ of \eqref{opt-P2} is given by the KKT conditions
\begin{align}
\nabla C(\overline{\bm p}_{g})- \widehat{\bm I}_g\overline{\bm \lambda} &= \bm 0, \\
\bm D_c^\top\overline{\bm\lambda} &= \bm 0, \\
-\bm D_c \overline{\bm \nu} + \widehat{\bm I}_g^\top\overline{\bm p}_{g}  - \overline{\bm p}_{\ell} - \overline{\bm\varphi}&= \bm 0.
\end{align} 
Since \eqref{opt-P2} is convex, the KKT conditions are also sufficient. This enables the primal-dual gradient method \cite{JokicLazarvandenBosch2009,Stegink.2015,Trip.2016} to be applied and results in \eqref{primal-dual-1}--\eqref{primal-dual-4}.
\end{proof}
Vector $\bm u_c$ is an additional control input and diagonal matrices $\bm \tau_g, \bm \tau_{\lambda}, \bm \tau_\nu > 0$ are controller gains, where small values for $\bm \tau_g, \bm \tau_{\lambda}, \bm \tau_\nu$ result in a faster convergence with larger transient amplitudes and vice versa.

Note that the distributed fashion of the controller \eqref{primal-dual-1}--\eqref{primal-dual-4} is provided by the fact that at each node $i \in \mathcal V_g$, local controller output $p_{g,i}$ depends only on variables that are calculated at node $i$ or at adjacent nodes.

A port-Hamiltonian representation of \eqref{primal-dual-1}--\eqref{primal-dual-4} is given by
\begin{align}
\dot{\bm x}_c = \underbrace{\begin{bmatrix} \bm 0 & \bm I & \bm 0  \\
	 -\bm I & \bm 0 & \bm D_{c} \\
	\bm 0 & -\bm D_{c}^\top& \bm 0 \end{bmatrix}}_{\bm J_c} \nabla H_c - \underbrace{\begin{bmatrix} \nabla C \\ -\bm \varphi \\  \bm 0 \end{bmatrix}}_{\bm r_c} 	+ \begin{bmatrix} \bm u_c \\  \bm p_{\ell} \\ \bm 0
		\end{bmatrix} \label{eq-controller-phs}
\end{align}
with the controller state 
$\bm x_c = \mathrm{col}\{ \bm \tau_g \bm p_g , \bm \tau_{\lambda}\bm \lambda,  \bm \tau_{\nu}\bm{\nu}\}$
and the controller Hamiltonian
\begin{align}
H_c(\bm x_c)=\frac 12 \bm x_c^\top \bm \tau_c^{-1}\bm x_c
\end{align}
with
\begin{align}
\bm \tau_c=\mathrm{diag}\{\bm \tau_g, \bm \tau_{\lambda},  \bm \tau_{\nu}\}.
\end{align}
\subsection{Closed-Loop System}
By choosing $\bm u_c = -\bm \omega_g$ as in \cite{Stegink.2017,Stegink.}, both plant \eqref{plant-phs} and controller \eqref{eq-controller-phs} can be interconnected in a power-preserving manner, leading to the closed-loop descriptor system:
\begin{align}
\bm E\dot{\bm x} = \left(\bm J - \bm R  \right) \nabla H - \bm r + \bm F \bm u \label{closed-loop}
\end{align}
with the closed-loop Hamiltonian $H(\bm x_p, \bm x_c) = H_p(\bm x_p) + H_c(\bm x_c)$ and
\begin{align}
\bm E &= \mathrm{diag}\{ \bm I_{3n_g+n+m_c+m}, \bm 0_{2n_\ell}\},\\
\bm J &= \begin{bmatrix}
\bm 0 & \bm I & \bm 0& \bm 0 & -\bm I & \bm 0 & \bm 0 & \bm 0  \\
-\bm I & \bm 0 & \bm D_{c} & \bm 0 & \bm 0 & \bm 0 & \bm 0 & \bm 0  \\
\bm 0 & -\bm D_{c}^{\top}   & \bm 0 & \bm 0 & \bm 0 & \bm 0 & \bm 0& \bm 0 \\
\bm 0 & \bm 0 & \bm 0 & \bm 0 & \bm D_{pg}^\top & \bm 0 & \bm D_{p\ell}^\top & \bm 0 \\
\bm I & \bm 0 & \bm 0 &  -\bm D_{pg} & \bm 0 & \bm 0 & \bm 0 & \bm 0  \\
\bm 0 & \bm 0 & \bm 0 &  \bm 0 & \bm 0 & \bm 0 & \bm 0 & \bm 0  \\
\bm 0 & \bm 0 & \bm 0 &  -\bm D_{p\ell} & \bm 0 & \bm 0 & \bm 0 & \bm 0 \\
\bm 0 & \bm 0 & \bm 0 & \bm 0 & \bm 0 & \bm 0 & \bm 0 & \bm 0  
\end{bmatrix}, \\
\bm R &= \mathrm{diag}\{\bm 0_{n_g+n+m_c}, \bm R_p \}, \\
\bm r &= \mathrm{col}\{\bm r_c, \bm r_p\}, \\
\bm F &= \begin{bmatrix} \bm 0 & \bm 0 & \bm 0 \\
\bm 0 & \bm 0 & {\bm I} \\
\bm 0 & \bm 0 & \bm 0 \\
\bm 0 & \bm 0 & \bm 0 \\
\bm 0 & \bm 0 & -\bm{\widehat I}_g \\
\bm{\hat \tau}_U & \bm 0 & \bm 0 \\
\bm 0 & \bm 0 & -\bm{\widehat I}_\ell \\
\bm 0 & -\bm I & \bm 0
\end{bmatrix}, \\
\bm u &= \mathrm{col}\{\bm U_f, \bm q_\ell , \bm p_\ell \}.
\end{align}
It is notable that for each equlibrium $\overline{\bm x}$, from \eqref{primal-dual-1} we have $\nabla C(\overline p_{gi})=\overline\lambda_i$ since $\overline{\omega}_i=0$. Moreover, from \eqref{primal-dual-4} it follows that $\bm{\overline\lambda} \in \ker{\bm D_c^\top}$, and as $\bm D_c$ is an incidence matrix, all elements of $\bm{\overline\lambda}$ must be equal. Summing up, it follows that all marginal prices are equal at steady state, which is the \emph{economic dispatch criterion} \cite{Doerfler2019}.

For the sake of brevity, we denote the co-state vector $\bm z=\nabla H$ and the dissipation vector $\mathcal R(\bm z,\bm x)=\bm R \bm z + \bm r$. Then, \eqref{closed-loop} is a port-Hamiltonian descriptor system with nonlinear dissipation \cite{vanderSchaft.2017}
\begin{align}
\bm E\dot{\bm x} = \bm J \bm z - \mathcal R(\bm z,\bm x) + \bm F \bm u. \label{closed-phs-1}
\end{align}
For a constant input $\overline{\bm{u}}$ the corresponding equilibrium $\overline{\bm x}$ is the solution of 
\begin{align}
{\bm 0} = \bm J \overline{\bm z} - \mathcal R(\overline{\bm z},\overline{\bm x}) + \bm F \overline{\bm u} \label{closed-phs-2}
\end{align}
where $\overline{\bm z}=\nabla H(\overline{\bm x})$. 
Since $H$ is a convex and nonnegative function, the shifted Hamiltonian \cite{vanderSchaft.2017}
\begin{align}
\overline H(\bm x)=H(\bm x)-\left( \bm x - \overline{\bm x} \right)^\top \nabla H(\overline{\bm x}) - H(\overline{\bm x}) \label{shifted-H}
\end{align}
is positive definite with minimum $\overline H(\overline{\bm x})=0$. Thus the shifted closed-loop dynamics, i.e. \eqref{closed-phs-1} minus \eqref{closed-phs-2}, can be expressed in terms of \eqref{shifted-H} as follows:
\begin{align}
\bm E\dot{\bm x} = \bm J \nabla \overline H - \left[\mathcal R(\bm z,\bm x)- \mathcal R(\overline{\bm z},\overline{\bm x})\right] + \bm F \left[\bm u - \overline{\bm u}\right].
\end{align}
As a result, stability of \eqref{closed-loop}  is given if the shifted passivity property \cite{vanderSchaft.2017}
\begin{align}
\left[\bm z - \overline{\bm z}\right]^\top \left[\mathcal R(\bm z,\bm x)- \mathcal R(\overline{\bm z},\overline{\bm x})\right] \geq 0 \label{eq-passivity}
\end{align}
is satisfied.
\section{Simulation}
We now validate the presented control approach by simulating a medium voltage power network as depicted in Fig. \ref{figure1}. The parameter values are partly based on those provided in \cite{Trip.2016} and are summarized in Tables \ref{tab2} and \ref{tab3}. However, deviating from \cite{Trip.2016},
\begin{enumerate}
	\item Line conductances $G_{ij}$ are nonzero. Without loss of generality, yet for sake of simplicity, we assume constant $R \slash X$ ratios $\eta$, i.e. $G_{ij}= - \eta\cdot B_{ij}$ for each line $(i,j)$,
	\item Generator reactances and line parameters are scaled down appropriately to suit a medium-voltage distribution grid with a base voltage of \SI{10}{kV}.
\end{enumerate} 
The simulations were carried out in Wolfram Mathematica 11.3. 
\begin{figure}[t]
 \begin{center}
  \includegraphics[width=0.40\textwidth]{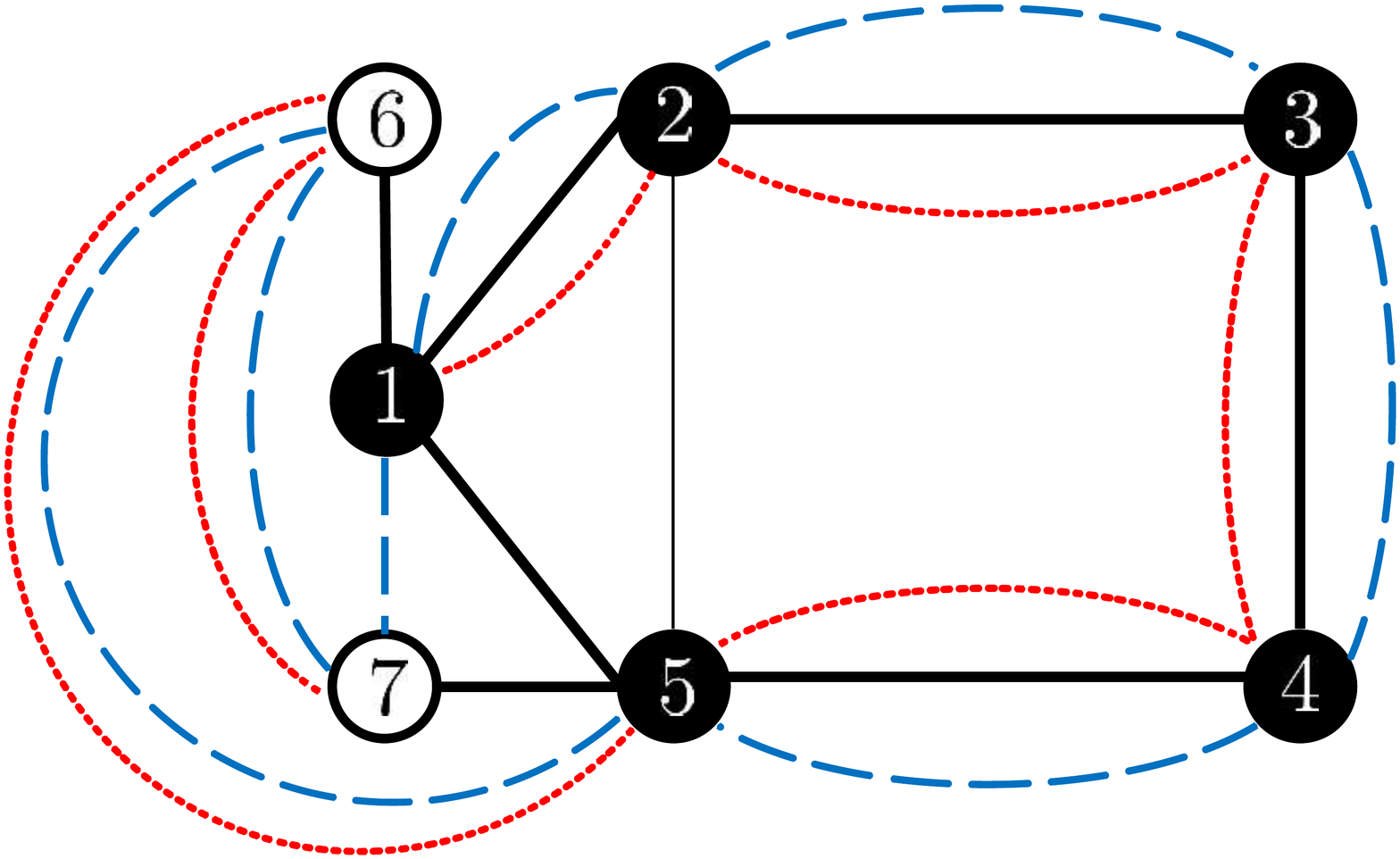}
  \caption{Simulation model with five generator nodes (1--5) and two load nodes (6 and 7). 
  	Solid black lines: physical interconnection via transmission lines. Dashed blue lines: Ring communication. Dotted red lines: Open ring communication.
  }
\label{figure1}
 \end{center}
\end{figure}
\begin{table}
	\renewcommand{\arraystretch}{1.3}
	\caption{Numerical values of the nodal parameters used in the simulations. The units of the parameters are given in p.u., except $\tau_{U,i}$ which is given in seconds.}
	\begin{center}
		\begin{tabular}{|l||l|l|l|l|l|l|l|}
			\hline
			$i$ 	& 1 & 2 & 3 & 4 & 5 & 6 & 7 \\
			\hline
			\hline
			$A_i$ & 1.6 & 1.2 & 1.4 & 1.4 & 1.5 & 1.3 & 1.3 \\
			$B_{ii}$ & -5.5 & -5.5 & -3.3 & -3.1 & -7.0 & -2.0 & -2.0 \\
			$M_i$ & 5.2 & 4.0 & 4.5 & 4.2 & 4.4 & \textendash & \textendash \\ 
			$X_{d,i}$ & 0.02 & 0.03 & 0.03 & 0.025 & 0.02 & \textendash & \textendash\\
			$X'_{d,i}$ & 0.004 & 0.006 & 0.005 & 0.005 & 0.003 & \textendash & \textendash\\
			$\tau_{U,i}$ & 6.45 & 7.7 & 8.3 & 7.0 & 7.36 & \textendash & \textendash \\
			\hline
		\end{tabular}
		\label{tab2}
	\end{center}
\end{table}

\begin{table}
	\renewcommand{\arraystretch}{1.3}
	\caption{Numerical values of the line parameters used in the simulations. The units of the parameters are given in p.u.}
	\begin{center}
		\begin{tabular}{|l|l|l|l|l|l|l|l|}
			\hline
			 $B_{12}$ & $B_{15}$ & $B_{16}$ & $B_{23}$ & $B_{25}$ & $B_{34}$ & $B_{45}$ & $B_{57}$ \\
			\hline
			\hline
			$1.27$ & 1.4 & 2.0 & 1.4 & 2.05 & 1.1 & 1.0 & 2.0 \\
			\hline
		\end{tabular}
		\label{tab3}
	\end{center}
\end{table}
\subsection{Parameterization of Input Signals and Cost Function}
In the following numerical examples, the power system is initially in a steady state with nominal frequency $\omega^n$ and constant power loads $p_{\ell,i}$. At time $ t= \SI{30}{\second}$ and $ t= \SI{60}{\second}$, 
a step load change of $+0.1$ p.u. occurs at load nodes $p_{\ell,6}$ and $p_{\ell,7}$, respectively. 

Moreover, the overall cost function is chosen to be
\begin{align}
C(\bm p_g)=\frac 12 \sum_{i =1}^5 \frac{1}{w_i}\cdot p_{g,i}^2
\end{align}
with weighting factors $w_1=\num{1}$, $w_2=\num{1.1}$, $w_3=1.2$, $w_4=1.3$, $w_5=1.4$. The choice of a weighted sum-of-squares is convenient \cite{Monshizadeh.} since the above-mentioned economic dispatch criterion leads to $\nabla C(\overline p_{g,i})=\overline p_{g,i} \slash w_i = \text{const.}$ for all $i=1,\ldots, 5$, i.e. active power sharing \cite{Schiffer.1}.
\subsection{Numerial Results}
\subsubsection{Effect of communication matrix}
First, we investigate the effect of different communication graphs on the performance of the distributed controllers. For the simulation, we employ four communication structures, namely
\begin{enumerate}[label = {[\alph*]}]
\item{a \textit{complete communication graph} (all-to-all communication)} 
\item{a communication graph \textit{identical to the physical topology of power system} (solid black lines in Fig.\ref{figure1})}
\item{a \textit{ring} (dashed blue lines in Fig.\ref{figure1})}
\item{an \textit{open ring} (dotted red lines in Fig. \ref{figure1})}
\end{enumerate}
 The $R \slash X$ ratio is set to one.
 
 As seen in Fig. \ref{figure2}, after a certain time of about 10 seconds, all nodal frequencies return to the nominal frequency, i.e. frequency regulation is achieved in all four cases.
 Moreover, as seen in Fig. \ref{figure3}, injections $p_{g,i}$ are equidistant for each post-fault equilibrium, i.e. active power sharing is also maintained in all four cases.
 The choice of $\bm D_c$ only affects the transient behavior where a sparse communication matrix results in a slighly bigger overshoot of $p_{g_i}$ after a step load change.
 \begin{figure}[t]
 	\centering
 	\subfloat{
 		\includegraphics[width=8.1 cm, height=3.24 cm]{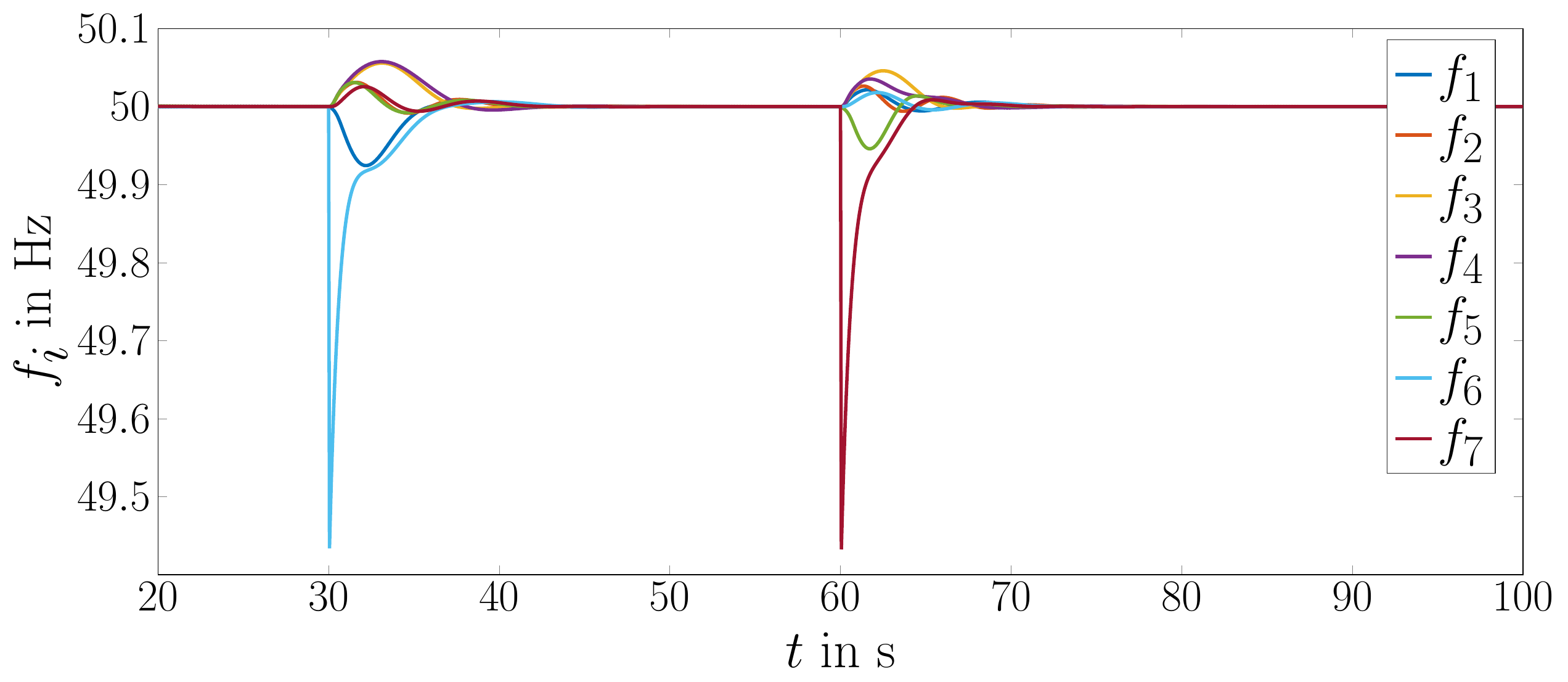}}
 	
 	\subfloat{
 		\includegraphics[width=8.1 cm,height = 3.24 cm]{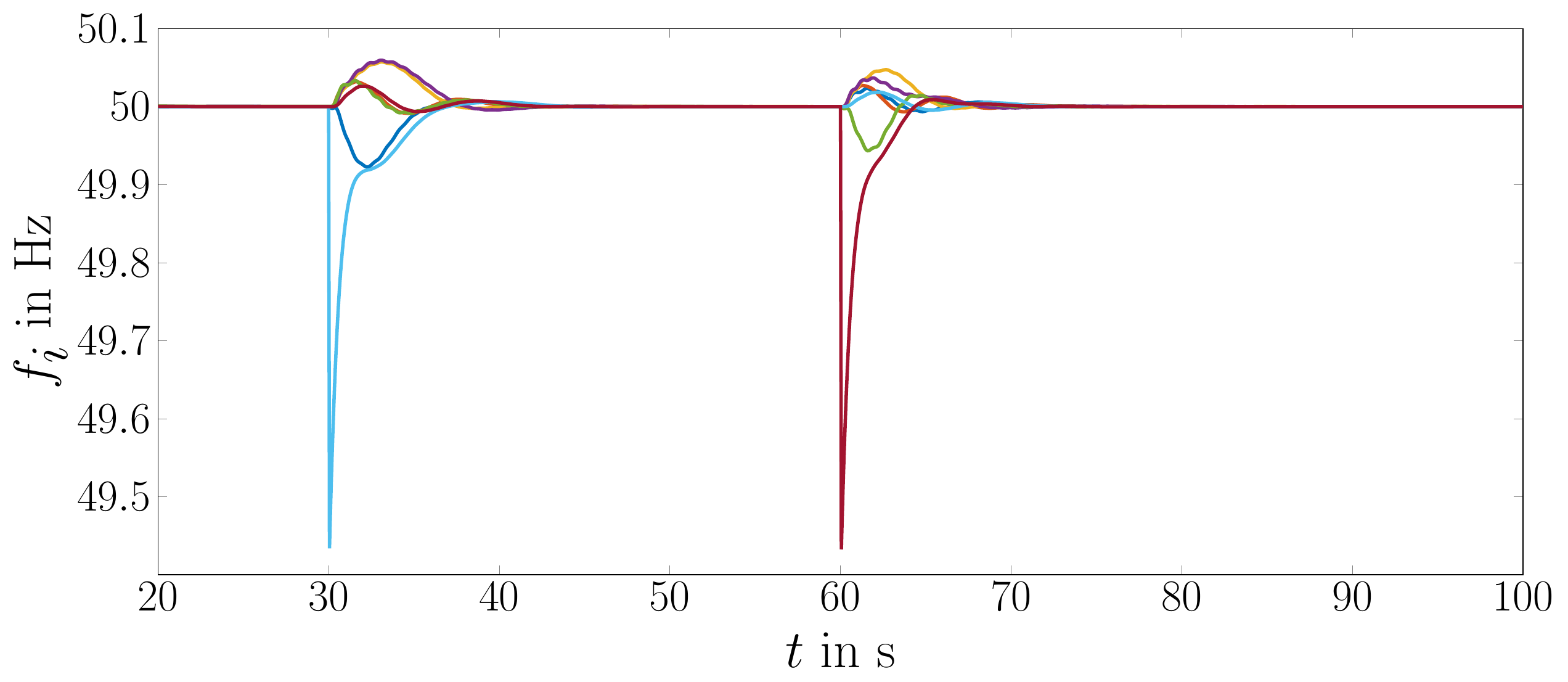}}
 	
 	\subfloat{
 		\includegraphics[width=8.1 cm,height = 3.24 cm]{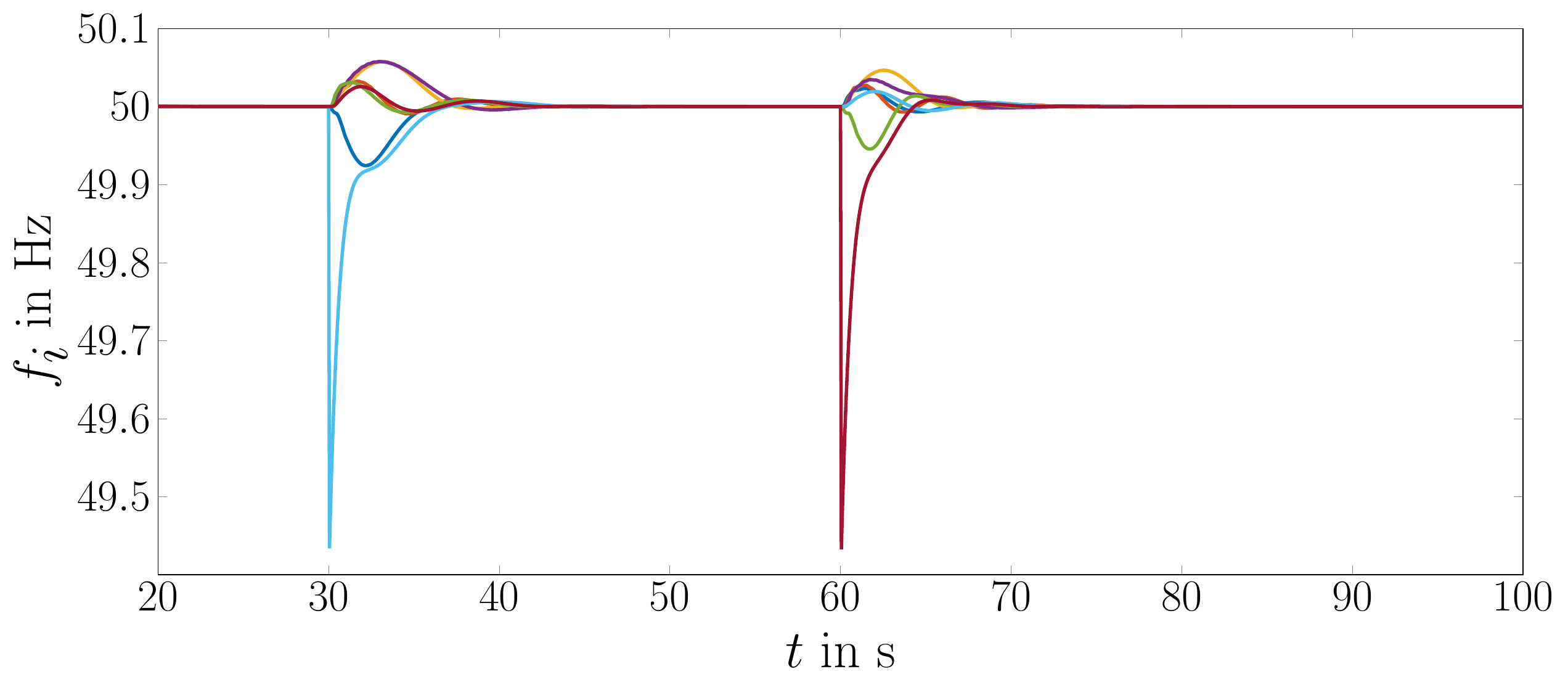}}
 	
 	\subfloat{
 		\includegraphics[width=8.1 cm,height = 3.24 cm]{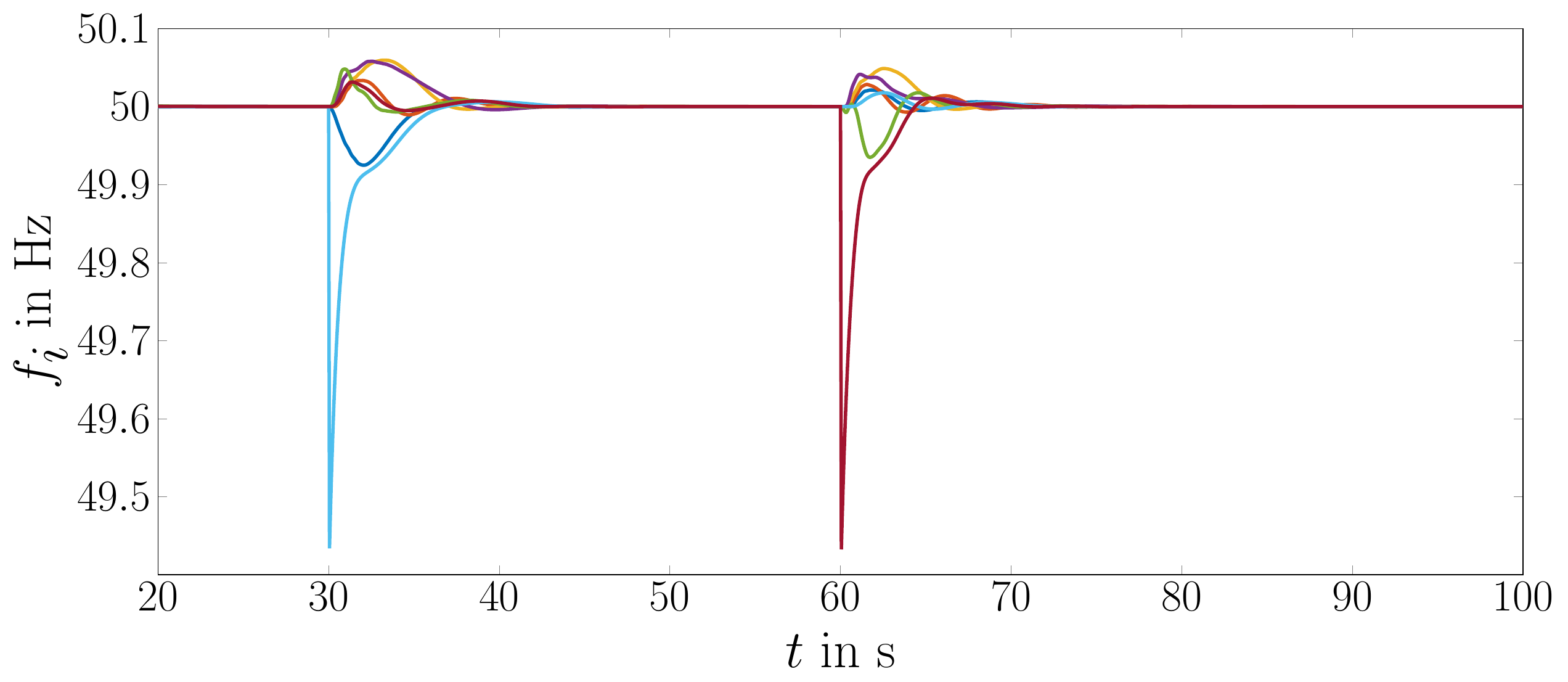}}
 	\caption{Frequency regulation after an increase in local demands $p_{\ell,6}$ and $p_{\ell,7}$ for 1)\,[a]--[d].}
 	\label{figure2}
 \end{figure}
 
\begin{figure}[t]
\centering
   \subfloat{
   \includegraphics[width=8.1 cm, height=3.24 cm]{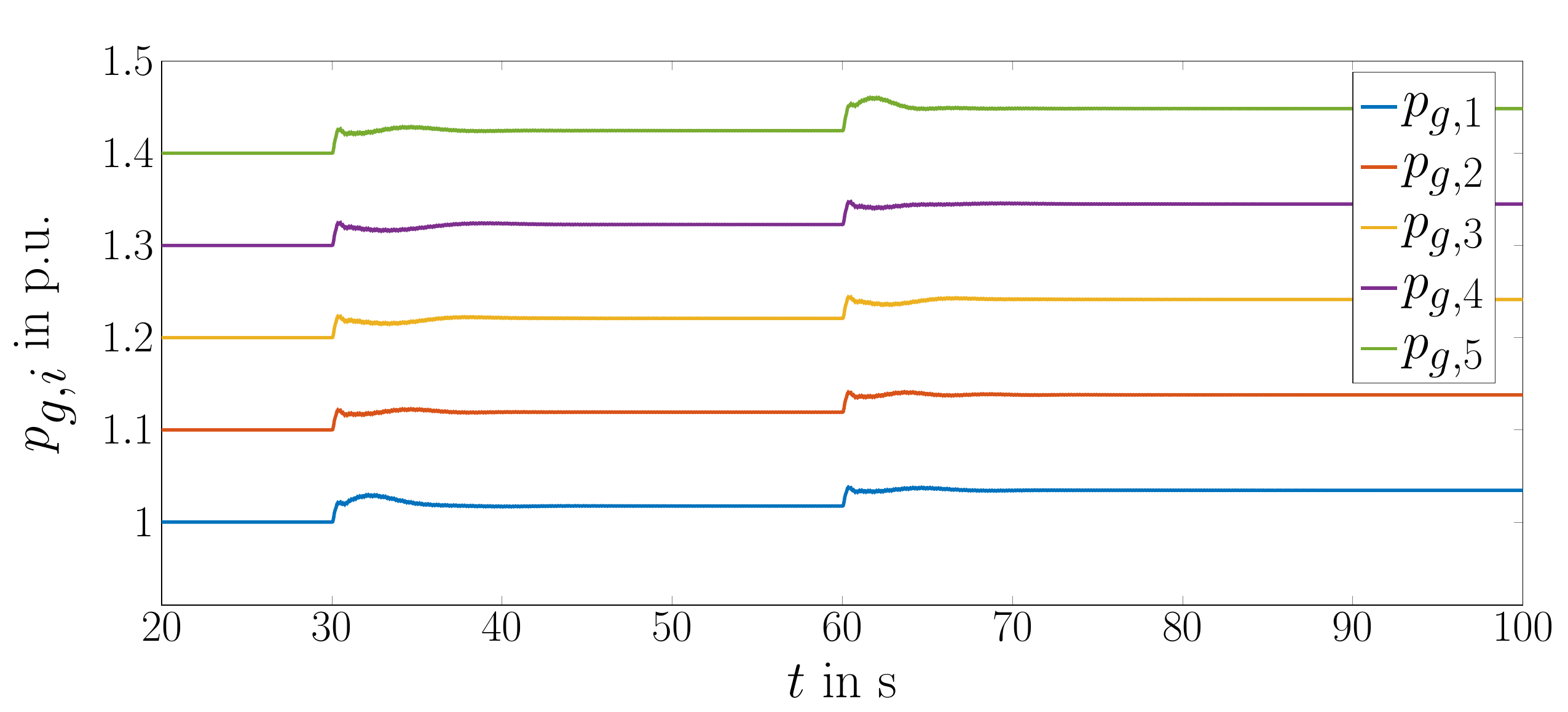}}

   \subfloat{
   \includegraphics[width=8.1 cm,height = 3.24 cm]{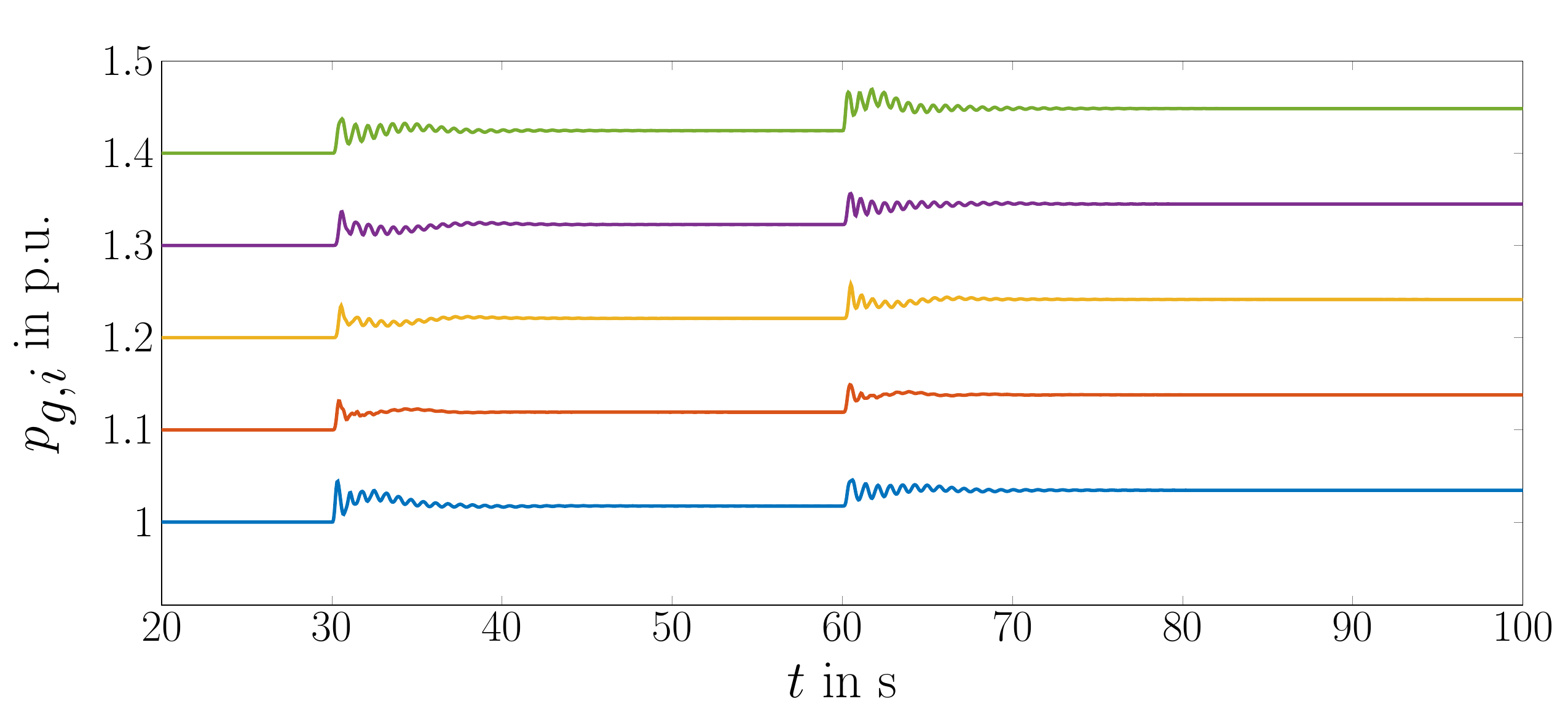}}

   \subfloat{
   \includegraphics[width=8.1 cm,height = 3.24 cm]{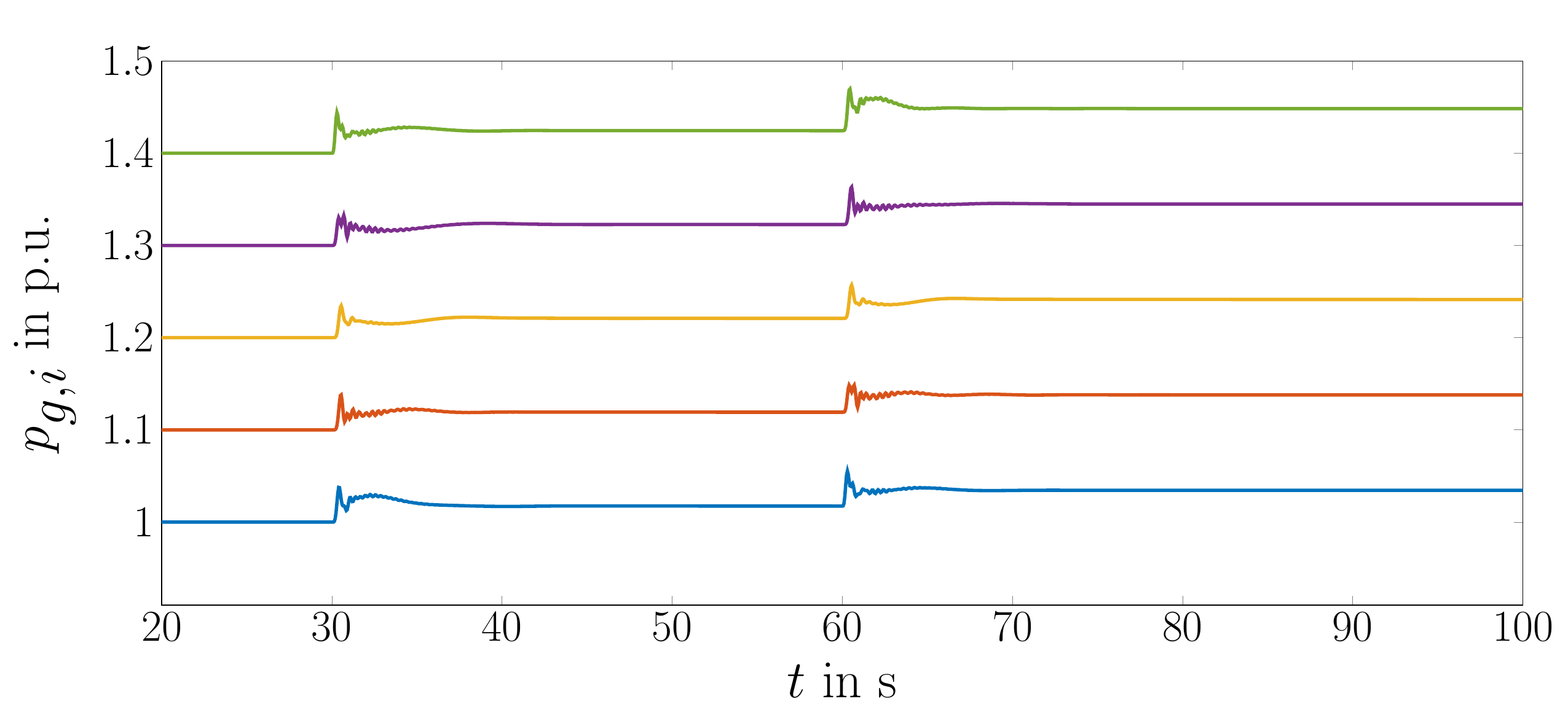}}

   \subfloat{
   \includegraphics[width=8.1 cm,height = 3.24 cm]{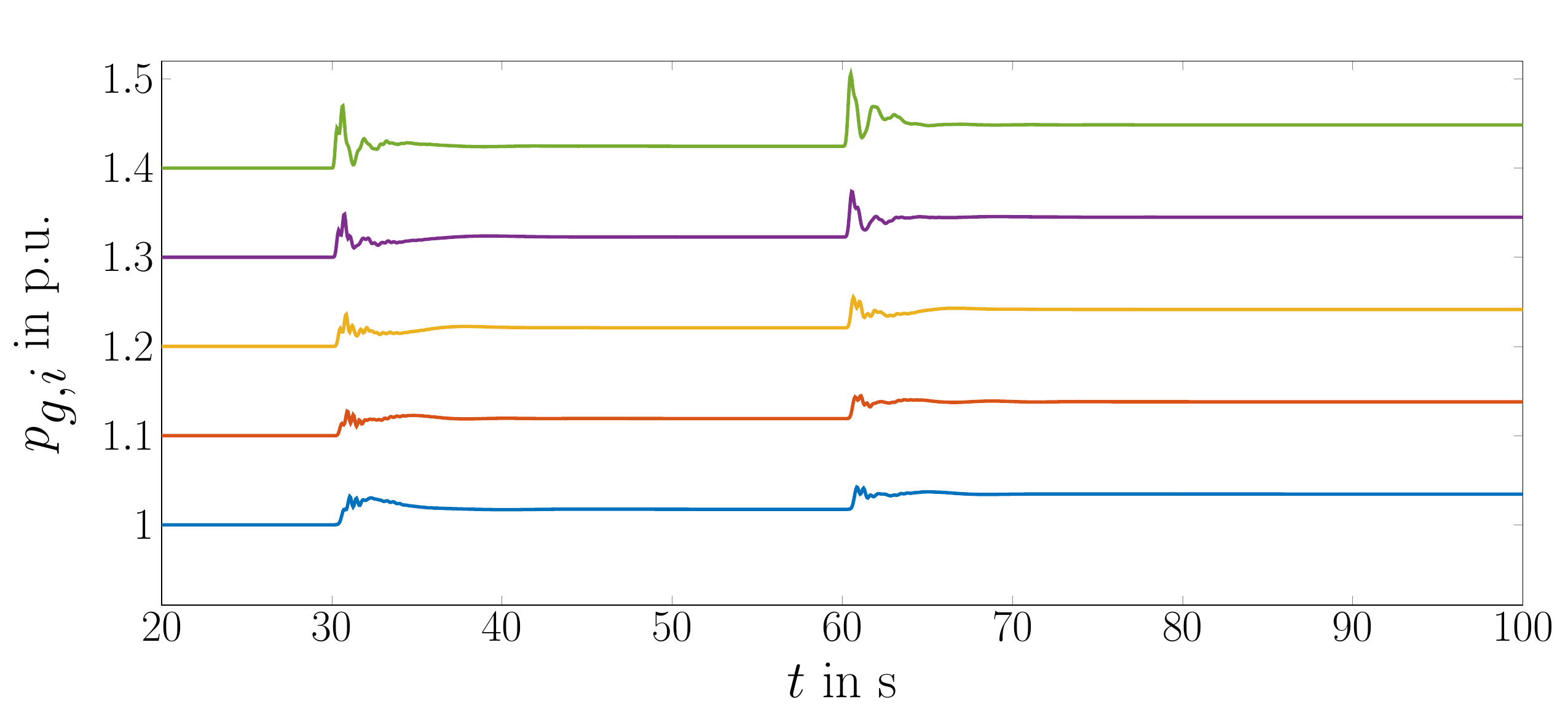}}
\caption{Optimal power injections after an increase in local demands $p_{\ell,6}$ and $p_{\ell,7}$ for 1)\,[a]--[d].}
\label{figure3}
\end{figure}
For reasons of limited space, the communication structure for the remaining simulations is always chosen to be identical to the physical topology (communication structure [b]).
\subsubsection{Comparison to \cite{Stegink.}}
 Fig. \ref{figure4} shows the price-based frequency control proposed in \cite{Stegink.} for $R\slash X$ being set to one. As can be seen in Fig. \ref{figure4}, this results in a steady-state deviation from the nominal frequency. Evidently, the closed-loop system begins to diverge from the nominal frequency well ahead of the load increments due to the unaccounted resistive losses generated in the lines.
 
Note that such a steady state deviation from nominal frequency always occurs in all simulations of \cite{Stegink.} with nonzero resistive losses, and that this deviation from nominal frequency increases when resistive losses increase.
 \begin{figure}[t]
\centering
   \includegraphics[width=8.1 cm,height = 3.24 cm]{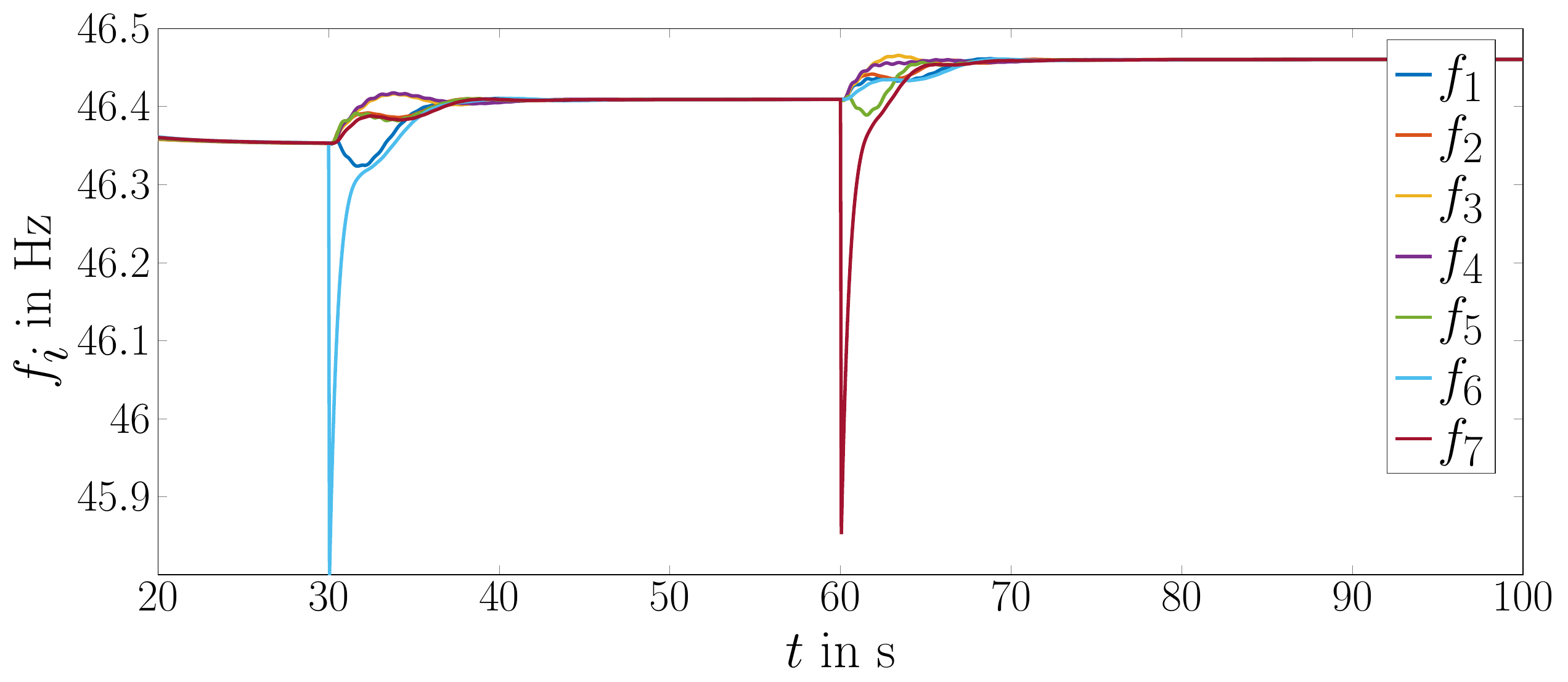}
\caption{Steady-state deviation from nominal frequency due to neglect of line conductances in the controller design \cite{Stegink.}.}
\label{figure4}
\end{figure}
\subsubsection{Effect of $R/X$ ratio on closed-loop stability}
In order to yield insight into stability of the closed loop, we carry out a numerical analysis of the shifted passivity property \eqref{eq-passivity} for various ${R\slash X}$ ratios. 

As apparent in Fig. \ref{figure5}, the (shifted) dissipations \eqref{eq-passivity} in the system \eqref{closed-loop} tend to increase with an increase in line losses. At the same time, however, they also have a more negative rate of change for more resistive lines, thereby being more vulnerable to a descent into instability. This can be observed for the case ${R\slash X = 3}$.  Therefore, we conclude that there exists a certain ${R\slash X}$ ratio, exceeding which the stability is no more given for the closed control loop.
\begin{figure}[t]
	   \includegraphics[width=8.1 cm, height=3.24 cm]{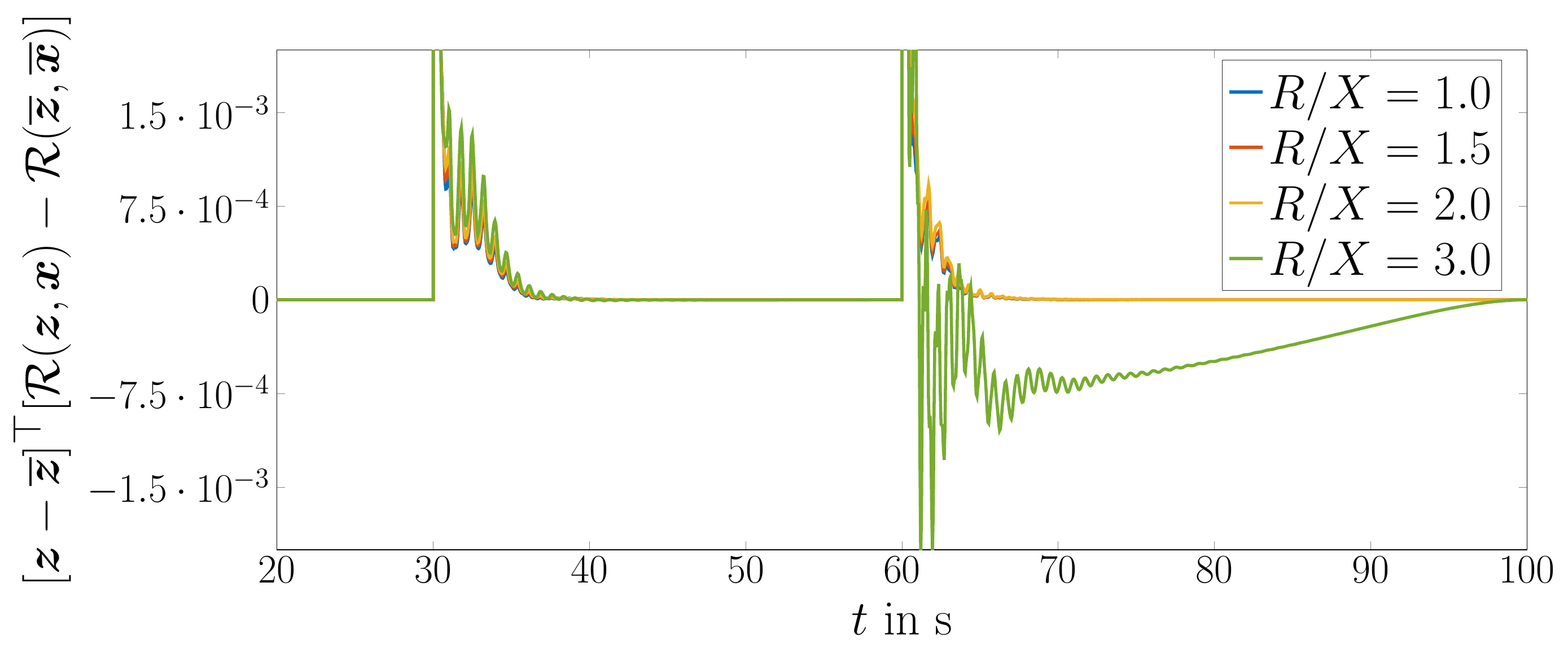}
\caption{Dissipations in the system for ascending $R\slash X$ ratios. The system becomes unstable for $R\slash X = 3$.}
\label{figure5}
\end{figure}  
\subsubsection{Clock drifts and communication failures}
A clock drift at a particular node would manifest itself in an incorrect measurement of the frequency deviation at that node, whereas a communication failure between two nodes would cut off the information flow, so that the neighbor-to-neighbor costs and line losses can no more be calculated. 

To simulate a clock drift of the controller at node 1, we assume that the measurement of $\omega_1$ at node 1 constantly deviates from the actual frequency by $\SI{-1}{\hertz}$.
The communication failure is modeled by eliminating the virtual power flow $\nu_{12}$ between the first and the second node. 

As depicted in Fig. \ref{figure6}, simulations show that the controller is robust in terms of clock drifts and is able to restore the nominal frequency after the step load changes. Notwithstanding, active power sharing can no longer be achieved. Moreover, the communication failure leads to a steady-state synchronous frequency deviant from the nominal frequency, see Fig. \ref{figure7}.
 \begin{figure}[t]
\centering
   \subfloat{
   \includegraphics[width=8.1 cm,height = 3.24 cm]{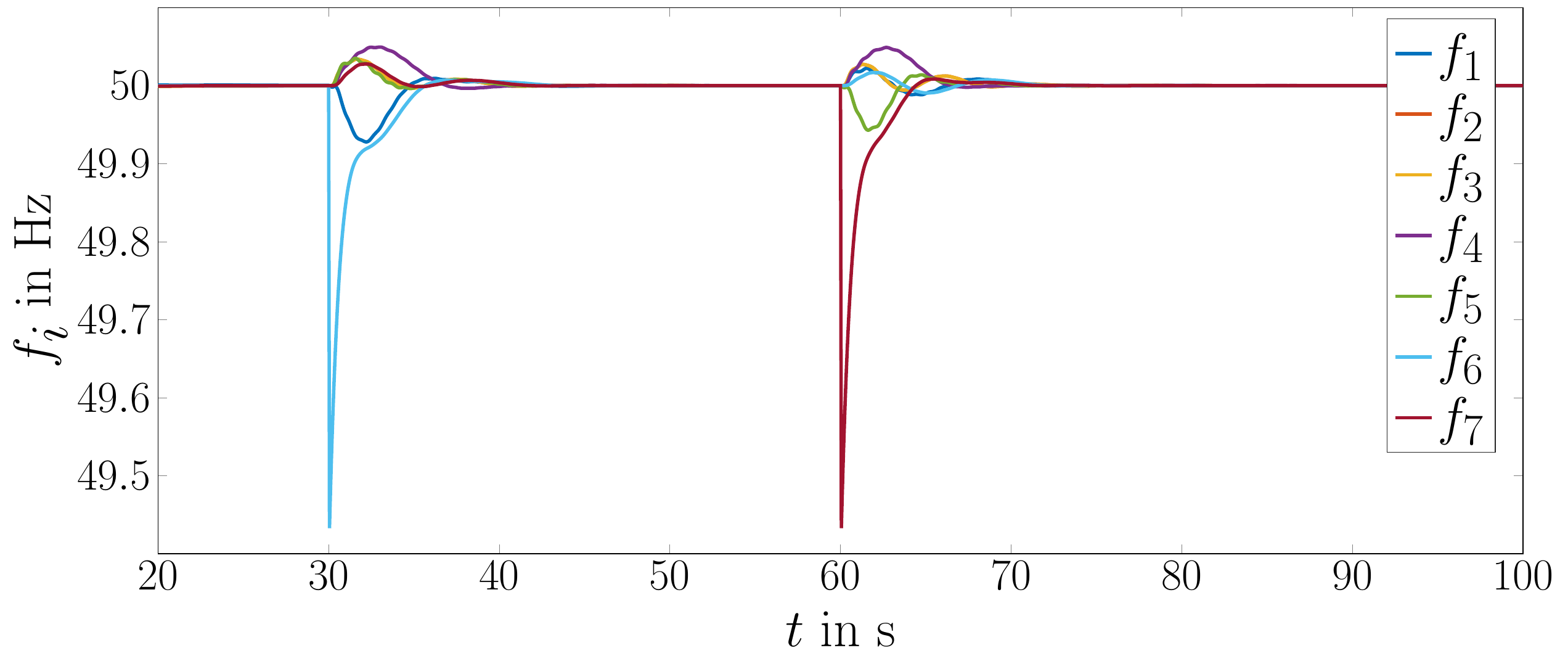}
}

   \subfloat{
   \includegraphics[width=8.1 cm,height = 3.24 cm]{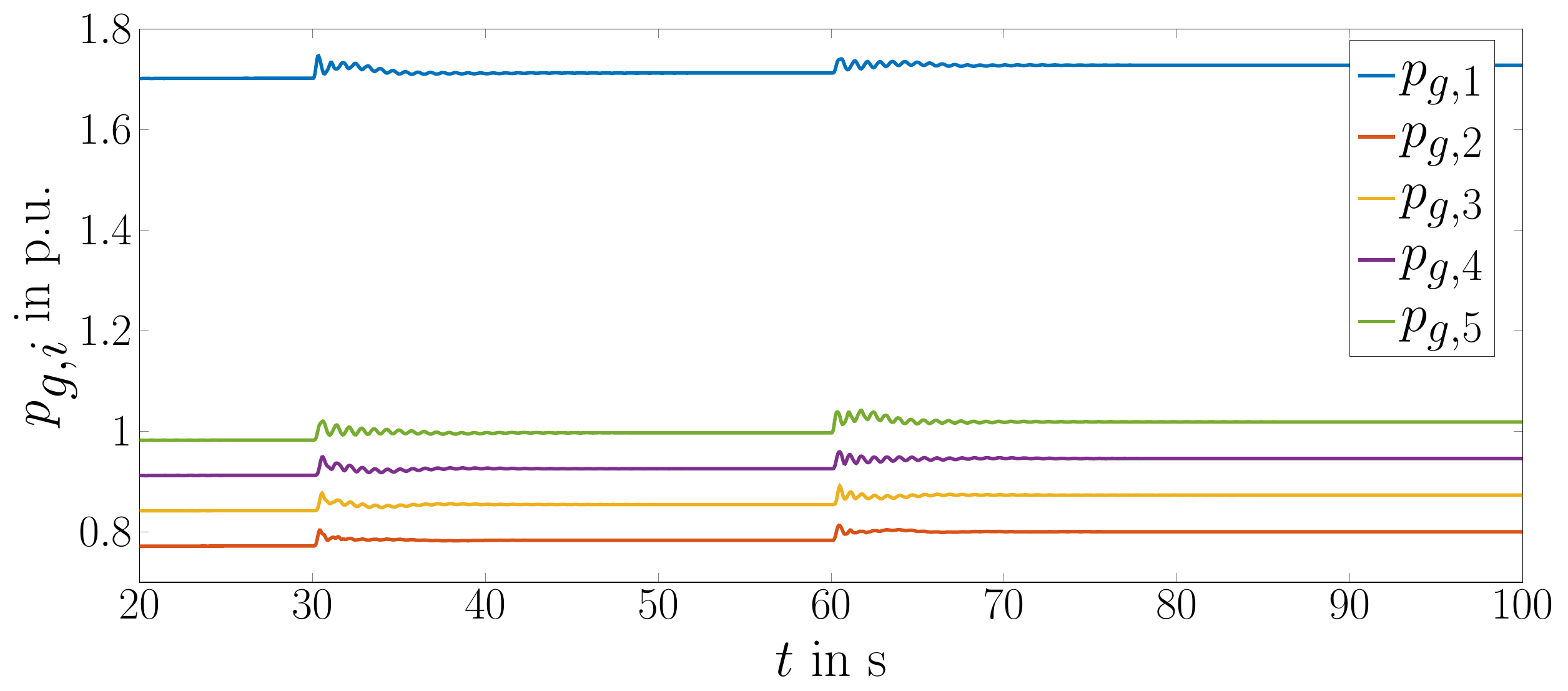}}
\caption{Performances of the distributed frequency control under a clock drift.}
\label{figure6}
\end{figure}

 \begin{figure}[t]
	\centering
		\includegraphics[width=8.1 cm, height=3.24 cm]{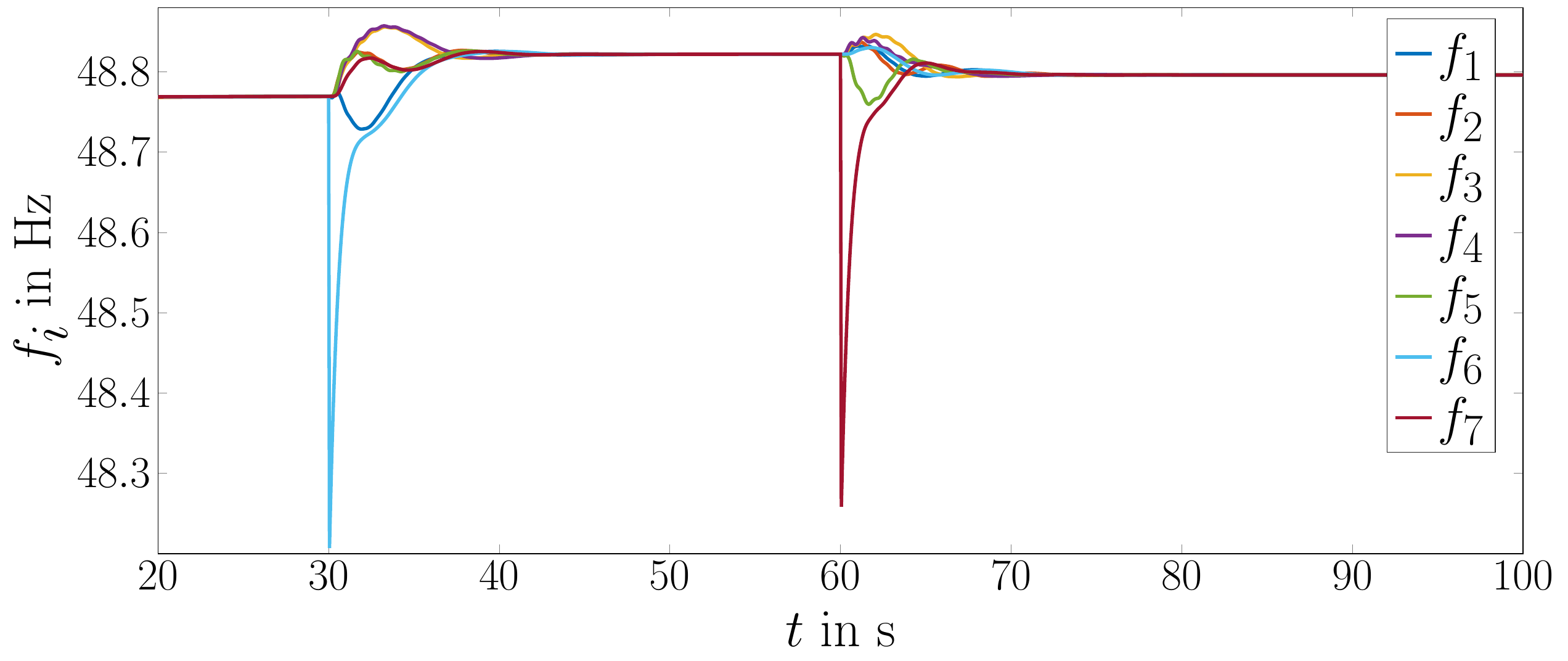}
	\caption{Steady-state deviation from nominal frequency due to communication failure.}	
	\label{figure7}
\end{figure}
\section{Summary and Outlook}
In this paper, we propose a price-based frequency control for lossy power grids that enables distributed communication and provides zero deviation from nominal frequency. 
The control method can be deployed for meshed as well as radial networks in distribution level power grids. It also takes into account load nodes with uncontrollable active power demands and results in a differential-algebraic nonlinear power grid model that can be represented as a port-Hamiltonian descriptor system. The passivity analysis based on simulations indicates a stable system up to a certain $R \slash X$ ratio in the lines.
Further research includes the additional consideration of power electronics-resourced interfaces and more rigorous stability results exploiting the (shifted) passivity property of the closed-loop system. Furthermore, nodal constraints such as generation limits and operational constraints for the transmission lines shall be included in the underlying optimization problem in order to always guarantee an operation that is in compliance with all technical regulations.
\newpage
\bibliography{Lukas_Stegink_Quellen} 
\bibliographystyle{unsrt}
\end{document}